\documentclass[11pt]{article}
\usepackage{geometry,a4wide}

\usepackage{mathrsfs}
\usepackage{lmodern}
\usepackage{amsmath,amssymb,amsthm,mathtools}
\usepackage{graphicx}
\usepackage{xcolor}
\usepackage[backend=bibtex,style=numeric,maxnames=99]{biblatex}
\usepackage{hyperref,cleveref}
\DeclareMathOperator{\CSP}{CSP}
\newcommand{\C}{\mathbb C}
\newcommand{\R}{\mathbb R}
\newcommand{\N}{\mathbb N}
\newcommand{\Z}{\mathbb Z}

\newcommand{\X}{\mathbf X}
\newcommand{\Y}{\mathbf Y}
\newcommand{\Lin}{\mathbf{Lin}}
\newcommand{\AIP}{\mathsf{AIP}}

\DeclareMathOperator{\Pol}{Pol}

\DeclareMathOperator{\SDP}{SDP}
\DeclareMathOperator{\ar}{ar}

\newcommand{\bone}{\mathbf{1}}  
\newcommand{\bzero}{\mathbf{0}}
\newcommand{\ba}{\mathbf{a}}

\newcommand{\bu}{\mathbf{u}}
\newcommand{\bx}{\mathbf{x}}

\newcommand{\by}{\mathbf{y}}

\newcommand{\be}{\mathbf{e}}

\newcommand{\minion}[1]{{\mathcal #1}}

\theoremstyle{plain}
\newtheorem{theorem}{Theorem}
\newtheorem{lemma}[theorem]{Lemma}
\newtheorem{proposition}[theorem]{Proposition}
\newtheorem{corollary}[theorem]{Corollary}

\theoremstyle{definition}
\newtheorem{definition}[theorem]{Definition}

\theoremstyle{remark}

\title{Linear equations mod $n$ are pseudo-telepathic}
\author{Lorenzo Ciardo\\
TU Graz\\
\texttt{lorenzo.ciardo@tugraz.at}
}
\date{}

\begin{document}
\maketitle

\begin{abstract}
We prove that the quantum monad in dimension
$2n$ admits no natural transformation to the polymorphism clone of
linear equations modulo $n$. Consequently, for every $n\geq 2$, there exists an unsatisfiable system of linear
equations over $\Z_n$ whose constraint system
game admits a perfect finite-dimensional quantum strategy. As a corollary, we completely characterise pseudo-telepathic constraint languages in finite dimension.
The proof combines a result of Harding, Jager, and Smith on group-valued
measures on subspaces of Hilbert spaces with the polymorphism-minion characterisation of
quantum pseudo-telepathy.
\end{abstract}

\section{Introduction}
\label{sec_intro}

In a linear constraint system game, two non-communicating provers (Alice and Bob) seek to convince a classical verifier that a system of linear equations modulo $n$ has a solution. Alice receives an equation and returns values for its variables; Bob receives a variable and returns a value for it. They win if Alice's assignment satisfies her equation and, whenever Bob's variable occurs in Alice's
equation, their answers are consistent. A strategy is perfect if it makes the players win with probability one for every pair of questions. A perfect classical strategy exists exactly when the linear system is satisfiable.
In a quantum strategy, the players may share a bipartite state $\psi\in H^{\otimes2}$ and base their answers on measurements of their part of the system, where $H$ is a nonzero finite-dimensional
Hilbert space~\cite{CleveM14}. 
We prove the following.

\begin{theorem}
\label{thm_main_pseudo_telepathy}
    For every $n\geq 2$, there exist unsatisfiable linear constraint systems modulo $n$ whose game admits a perfect quantum strategy over $\C^{2n}$.
\end{theorem}

For $n=2$,~\Cref{thm_main_pseudo_telepathy} is well known and follows from four-dimensional parity proofs of the Bell--Kochen--Specker theorem~\cite{bell1966problem,kochen1967problem}, such as the 18-vector construction of~\cite{cabello1996bell}, or from the Mermin--Peres magic square~\cite{mermin1990simple,peres1990incompatible} and its nonlocal-game formulation~\cite{cabello2001all,aravind2004quantum}.
It was noted  by Qassim and Wallman~\cite{qassim2020classical} that the magic-square construction works for every even $n$, with perfect quantum strategies given by tensor products of generalised Pauli operators.
Moreover, the existence of unsatisfiable linear systems modulo $n=3$ and $n=5$ admitting perfect finite-dimensional quantum strategies has recently been announced by van Dobben de Bruyn and Roberson~\cite{vDobben2026magic}.
In addition,~\Cref{thm_main_pseudo_telepathy} strengthens the analogous result proved in~\cite{slofstra2024operator} by Slofstra and Zhang, which concerns commuting-operator strategies as opposed to finite-dimensional tensor-product strategies.
Indeed, every tensor-product strategy is in particular a commuting-operator strategy,
while games admitting a perfect commuting-operator strategy but no perfect tensor-product  strategy are known by Slofstra's answer~\cite{slofstra2020tsirelson} to Tsirelson's question~\cite{tsirelson2006bell}.

\section{Structures}
\label{sec_structures}
Our proof follows the same approach adopted in
\cite[Proposition~15]{ciardo_quantum_minion} to  show that the quantum monad of Abramsky, Barbosa, de Silva, and Zapata~\cite{abramsky2017quantum}, interpreted as a minion, does not admit a natural transformation to the dictator clone in any dimension $d\geq 3$. The key difference is that, here, we replace the use of Gleason's theorem~\cite{gleason1975measures} therein with a group-valued version due to Harding, Jager, and
Smith~\cite{harding2005group}.

First, let us introduce some    terminology. We let $\N$ be the set $\{1,2,\dots\}$ and, for $r\in\N$, we write $[r]=\{1,\ldots,r\}$. We view tuples $\ba=(a_1,\dots,a_r)$ as column vectors. $\bone$ and $\bzero$ denote the all-one and all-zero vectors, respectively.
A \emph{signature} $\sigma$ is a (finite) set of relation symbols $R$, each with its \emph{arity} $\ar(R)\in\N$.
A \textit{$\sigma$-structure} $\Y$ consists of a (finite) set $Y$ (the \emph{domain} of $\Y$) and a relation $R^\Y\subseteq Y^{\ar(R)}$ for each symbol $R\in\sigma$. 
A
\emph{homomorphism} $h:\X\to\Y$ between    two $\sigma$-structures $\X,\Y$ is a map $h:X\to Y$ such that
$\bx\in R^{\X}$ implies $h(\bx)\in R^{\Y}$
for every relation symbol $R\in\sigma$, where $h(\bx)$ is applied entrywise. We write $\X\to\Y$ to denote that a homomorphism from $\X$ to $\Y$ exists.
The \textit{constraints} of $\X$ are the pairs $c=(R,\bx)$ with $R\in\sigma$ and $\bx\in R^{\X}$.
In the \textit{$\X,\Y$ homomorphism game} (generalising the linear constraint system game discussed in~\Cref{sec_intro}), Alice and Bob seek to convince a classical verifier that $\X\to\Y$.
The verifier sends Alice a constraint $c=(R,\bx)$ of $\X$ and Bob a variable $x\in X$. Alice returns a tuple $\by\in Y^{\ar(R)}$, and Bob   returns a value $y\in Y$. They win if $\mathbf y\in R^{\Y}$ and $y_i=y$ for every coordinate $i$ such that $x_i=x$.
A strategy for Alice and Bob is \emph{perfect} if it makes them win regardless of the verifier's questions. A perfect classical strategy exists exactly when $\X\to\Y$. A perfect quantum strategy allows the players to share a bipartite state $\psi\in H^{\otimes 2}$ for some nonzero finite-dimensional Hilbert space $H$, and base their answers on measurements of their part of $\psi$. 
We shall assume that $\psi$ is the maximally entangled state in $H^{\otimes 2}$, which can be done without loss of generality by~\cite{abramsky2017quantum}. 
We call a $\sigma$-structure $\Y$ \emph{pseudo-telepathic over $H$} if there is some $\sigma$-structure $\X$ for which the $\X,\Y$ homomorphism game has a perfect finite-dimensional quantum strategy over $H$ but no perfect classical strategy~\cite{galliard2002pseudo,brassard2005quantum}.

To prove~\Cref{thm_main_pseudo_telepathy}, it will suffice to consider linear equations with a particularly simple shape, encoded by a finite reduct of the structure  whose relations are the affine hyperplanes of $\Z_n^r$ for any $r\in\N$.
For an integer $n\geq2$, we let $\Lin_n$ be the
structure with domain $\Z_n$ having a unary relation $C_0^{\Lin_n}=\{0\}$ and, for each $b\in\Z_n$, a ternary relation $R_b^{\Lin_n}=\{(y_1,y_2,y_3)\in\Z_n^3:y_1+y_2+y_3=b\}$.
An instance of the \textit{constraint satisfaction problem (CSP)} parameterised by $\Lin_n$ is a system of linear equations modulo $n$, each having three variables and all coefficients equal to one;
variables may occur more than once in an equation and may be constrained to equal zero.\footnote{In fact, we could consider even simpler types of equations, by taking only the ternary symbols $R_0$ and $R_1$. We add the other symbols to make the argument slightly cleaner.} We shall prove the following result, which immediately gives~\Cref{thm_main_pseudo_telepathy}.

\begin{theorem}\label{thm_ain}
For every $n\geq2$,  the structure $\Lin_n$ is pseudo-telepathic over $\C^{2n}$.
\end{theorem}

\section{Polymorphisms and minions}
\label{sec_minions_pseudo_telepathy}
We now introduce the notions of minions and polymorphisms from the theory of (promise) CSP complexity; see e.g.~\cite{BBKO21}.

\begin{definition}
A \emph{minion} $\minion M$ is a functor from the skeleton category of non-empty finite sets to
the category of non-empty sets. In other words, $\minion M$ is the disjoint union of non-empty sets $\minion M^{(\ell)}$ for $\ell\in\N$, equipped with operations $(\cdot)_{/\pi}:\minion M^{(\ell)}\to \minion M^{(\ell')}$ (known as \emph{minor operations}) for each pair of integers $\ell,\ell'\in\N$ and each map
$\pi:[\ell]\to [\ell']$, which must satisfy the following two conditions:
\begin{itemize}
    \item $M_{/\operatorname{id}}=M$ 
    for each $\ell\in\N$ and $M\in\minion M^{(\ell)}$, where $\operatorname{id}$ is the identity map on the set $[\ell]$;
    \item $(M_{/\pi})_{/\tilde\pi}=M_{/\tilde\pi\circ\pi}$
     for each $\ell,\ell',\ell''\in\N$, $M\in\minion M^{(\ell)}$, $\pi:[\ell]\to[\ell']$,
    and $\tilde\pi:[\ell']\to[\ell'']$.
\end{itemize}
\end{definition}

\begin{definition}
A \emph{minion homomorphism} $\xi:\minion M\to\minion N$ is a natural transformation from $\minion M$ to $\minion N$. In other words, $\xi$ is a map from the underlying set of $\minion M$ to the underlying set of $\minion N$ that
\begin{itemize}
    \item preserves the arities---i.e., $\xi(M)\in\minion N^{(\ell)}$ if $M\in\minion M^{(\ell)}$;
    \item preserves the minors---i.e., $\xi(M_{/\pi})=\xi(M)_{/\pi}$ if $M\in\minion M^{(\ell)}$ and  $\pi:[\ell]\to[\ell']$.
\end{itemize}
We write $\minion M\to\minion N$ when such a $\xi$ exists; otherwise, we write $\minion M\not\to\minion N$. We also let a \textit{minion isomorphism} be a bijective minion homomorphism.
\end{definition}

Let $\Y$ be a $\sigma$-structure. 
An $r$-ary \emph{polymorphism} of $\Y$ is a function $f:Y^r\to Y$
that preserves every relation of $\Y$ in the following sense: for any $k$-ary $R\in\sigma$
and any $\mathbf y^{(1)},\ldots,\mathbf y^{(r)}\in R^{\Y}$, the tuple
\begin{align*}
 (f(y_1^{(1)},\dots,y_1^{(r)}),\dots,
       f(y_k^{(1)},\dots, y_k^{(r)}))
\end{align*}
belongs to $R^{\Y}$. The polymorphisms of $\Y$ form a minion $\Pol(\Y)$ with the minor operations given by
\begin{align*}
 f_{/\pi}(z_1,\dots,z_s)
   =f(z_{\pi(1)},\dots,z_{\pi(r)}).
\end{align*}

Fix a finite-dimensional Hilbert space $H$, and let $L_H$ be the set of linear subspaces of $H$.
We next define the \textit{quantum monad} $\minion Q_H$ introduced in~\cite{abramsky2017quantum}; following~\cite{ciardo_quantum_minion}, we shall conveniently view it as a minion.

\begin{definition}
For a finite-dimensional Hilbert space $H$, the \emph{quantum monad} $\minion Q_H$ is the minion whose $r$-ary set is
\begin{align*}
 \minion Q_H^{(r)}=
 \left\{(U_1,\dots,U_r)\in L_H^r:
    U_i\perp U_j\ (i\ne j),\quad \bigoplus_{i\in[r]}U_i=H\right\}.
\end{align*}
For $\pi:[r]\to[s]$, the corresponding minor operation is defined as follows:
\begin{align*}
 (U_1,\ldots,U_r)_{/\pi}
   =\left(\bigoplus_{i\in\pi^{-1}(1)}U_i,\dots,
           \bigoplus_{i\in\pi^{-1}(s)}U_i\right),
\end{align*}
where an empty sum is the space $\{0\}$.
\end{definition}

We use the following characterisation of pseudo-telepathy.

\begin{theorem}[{\cite[Theorem~6]{ciardo_quantum_minion}}]
\label{thm_characterisation_PT}
A structure $\Y$ is pseudo-telepathic over a finite-dimensional Hilbert space
$H$ if, and only if,
$\minion Q_H\not\to\Pol(\Y)$.
\end{theorem}

\section{Polymorphisms of linear equations}
We now describe the polymorphism minion of $\Lin_n$. All arithmetic in
this section is in $\Z_n$.
\begin{definition}
For $n\geq2$, define the minion $\AIP_n$ by
 $\AIP_n^{(r)}
   =\left\{\ba\in\Z_n^r:
                  \ba^\top\bone=1\right\}$,
with minor operations
$\ba_{/\pi}
   =\left(\sum_{i\in\pi^{-1}(1)}a_i,\dots,
           \sum_{i\in\pi^{-1}(s)}a_i\right)$.
\end{definition}
The following fact was essentially proved in~\cite{BBKO21} to capture the power of the \textit{affine integer programming} algorithm of~\cite{BG18} via polymorphism identities (whence the name of the minion).
We give a short self-contained proof for completeness.  

\begin{proposition}[\protect{cf.~\cite[\S7.3]{BBKO21}}]
\label{prop_pol}
For every $n\geq2$, the minions $\Pol(\Lin_n)$ and
$\AIP_n$ are isomorphic. 
\end{proposition}

\begin{proof}
Consider the map $\xi: \AIP_n\to\Pol(\Lin_n)$ defined as follows: for each $r\in\N$, each $\ba\in\AIP_n^{(r)}$, and each $\bx\in\Z_n^r$, $\xi(\ba)(\bx)=\ba^\top\bx=\sum_{i\in[r]}a_ix_i$. We claim that $\xi$ is well defined and it is a minion homomorphism.

First, we show that $\xi(\ba)$ is a polymorphism of $\Lin_n$. Note that $\xi(\ba)(\bzero)=0$, so $\xi(\ba)$ preserves $C_0$. Moreover, let $M$ be an $r\times 3$ matrix all of whose rows sum up to some $b\in\Z_n$. Then the application of $\xi(\ba)$ to the columns of $M$ yields $(\ba^\top M\be_1,\ba^\top M\be_2,\ba^\top M\be_3)$, where $\be_i$ is the $i$-th standard unit vector. The sum of these values is $\ba^\top M\bone=b\ba^\top\bone=b$, so $\xi(\ba)$ preserves $R_b$ and is thus a polymorphism.

Take a map $\pi:[r]\to[s]$, and let $P_\pi$ be the $s\times r$ matrix whose $i$-th column is $\be_{\pi(i)}$. Note that
\begin{align*}
    \xi(\ba_{/\pi})(\bx)
    =
    (P_\pi\ba)^\top\bx=\ba^\top P_{\pi}^\top\bx=
    \xi(\ba)_{/\pi}(\bx),
\end{align*}
so $\xi$ preserves minors and, thus, is a minion homomorphism.

Let now $f$ be an $r$-ary polymorphism of $\Lin_n$. Preservation of $C_0$ gives $f(\bzero)=0$, and preservation of $R_0$ applied to $(\bx,\bzero,-\bx)$ gives $f(-\bx)=-f(\bx)$. Hence, preservation of $R_0$ applied to $(\bx,\by,-\bx-\by)$ gives $f(\bx+\by)=f(\bx)+f(\by)$, so $f$ is a group homomorphism $\Z_n^r\to\Z_n$. Since the standard unit vectors generate $\Z_n^r$, it follows that 
    $f(\bx)=\ba_f^\top\bx$ for each $\bx$, where $\ba_f=(f(\be_1),\dots,f(\be_r))$. Clearly, $\ba_f$ is uniquely determined. Furthermore, preservation of $R_1$ on $(\bone,\bzero,\bzero)$ gives $f(\bone)=1$, so $\ba_f^\top\bone=1$. Hence, $\ba_f\in\AIP_n^{(r)}$. The map $\vartheta:\Pol(\Lin_n)\to\AIP_n$ given by $f\mapsto \ba_f$  is the inverse of $\xi$, so $\xi$ is a minion isomorphism.
\end{proof}

\section{A group-valued Gleason theorem and pseudo-telepathy}

We next connect the quantum monad with the minions $\AIP_n$ seen in the previous section.

\begin{definition}
Let $H$ be a finite-dimensional Hilbert space and let $A$ be an Abelian group.
An \emph{$A$-measure on the subspaces of $H$} is a
map $\mu:L_H\to A$ such that
\begin{align*}
 \mu(U\oplus V)=\mu(U)+\mu(V)\qquad\text{whenever }U\perp V.
\end{align*}
\end{definition}

The next result uses the same idea as the one in the proof of~\cite[Proposition~15]{ciardo_quantum_minion} (and could be phrased in the more general setting of \textit{linear minions}~\cite{cz23soda:minions}).
\begin{proposition}
\label{prop_measure}
For every $n\geq2$ and every finite-dimensional Hilbert space $H$, a minion homomorphism $\minion Q_H\to\AIP_n$ exists
if, and only if, there is a $\Z_n$-measure $\mu$ on the subspaces of $H$ satisfying
$\mu(H)=1$.
\end{proposition}

\begin{proof}
Let $\xi:\minion Q_H\to\AIP_n$ be a minion homomorphism, and define $\mu(U)=\xi((U,U^\perp))_1$ for each $U\in L_H$. For any $U\perp V\in L_H$, write $\xi((U,V,(U\oplus V)^\perp))=(a,b,1-a-b)$. Since $\xi$ is a minion homomorphism, we have that $\mu(U)=a$, $\mu(V)=b$, and $\mu(U\oplus V)=a+b$. Hence, $\mu$ is a $\Z_n$-measure on the subspaces of $H$. Moreover, $\xi(H)=1$ and, thus, $\xi((H,\{0\}))=(1,0)$, which means that $\mu(H)=1$.

Conversely, suppose $\mu$ is a $\Z_n$-measure  on the subspaces of $H$ satisfying $\mu(H)=1$. For a tuple $\bu=(U_1,\dots,U_r)\in\minion Q_H^{(r)}$, we let $\xi(\bu)=(\mu(U_1),\dots,\mu(U_r))\in\Z_n^r$. Since the $U_i$'s are mutually orthogonal, the fact that $\mu$ is a measure on the subspaces of $H$ implies that 
    \begin{align*}\sum_{i\in[r]}\mu(U_i)=\mu\left(\bigoplus_{i\in[r]}U_i\right)=\mu(H)=1,
    \end{align*}
    so $\xi(\bu)\in\AIP_n^{(r)}$. To show that $\xi$ is a minion homomorphism, take a map $\pi:[r]\to[s]$, and observe that
    \begin{align*}
        \xi(\bu_{/\pi})
        &=
        \xi\left(\bigoplus_{i\in\pi^{-1}(1)}U_i,\dots,\bigoplus_{i\in\pi^{-1}(s)}U_i\right)
        =
        \left(\mu\left(\bigoplus_{i\in\pi^{-1}(1)}U_i\right),\dots,\mu\left(\bigoplus_{i\in\pi^{-1}(s)}U_i\right)\right)\\
        &=\left(\sum_{i\in\pi^{-1}(1)}\mu(U_i),\dots,\sum_{i\in\pi^{-1}(s)}\mu(U_i)\right)
        =
        \xi(\bu)_{/\pi},
    \end{align*}
    where the third equality uses that subspaces with distinct indices are mutually orthogonal and, thus, $\mu$ behaves additively on them. 
\end{proof}

The next result can be seen as an analogue of Gleason's theorem~\cite{gleason1975measures} for group-valued measures on subspaces of Hilbert spaces.

\begin{theorem}[{\cite[Lemma~2]{harding2005group}}]
\label{lem_hjs}
For every $n\geq2$, every $\Z_n$-measure $\mu$ on the subspaces of $\R^{2n}$ satisfies $\mu(\R^{2n})=0$.
\end{theorem}

We shall need the following observation.
\begin{lemma}
\label{lem_complexification}
For every $n\in\N$, there exists a minion homomorphism $\minion Q_{\R^n}\to\minion Q_{\C^n}$.    
\end{lemma}
\begin{proof}
    For any $U\in L_{\R^n}$, consider its complexification $U^{\C}=U\otimes_{\R}\C\in L_{\C^n}$. Let $\xi:\minion Q_{\R^n}\to\minion Q_{\C^n}$ be the map defined by $(U_1,\dots,U_r)\mapsto (U_1^\C,\dots,U_r^\C)$. We claim that $\xi$ is a minion homomorphism. First, for $i\neq j$, since $U_i\perp U_j$, we have that $U_i^\C\perp U_j^\C$. Moreover,
    \begin{align*}
        \bigoplus_{i\in[r]}U_i^\C=\bigoplus_{i\in[r]}(U_i\otimes_{\R}\C)=(\bigoplus_{i\in[r]}U_i)\otimes_{\R}\C=\R^n\otimes_\R\C=\C^n,
        \end{align*}
        so $\xi$ is a well-defined map from $\minion Q_{\R^n}$ to $\minion Q_{\C^n}$. The fact that $\xi$ preserves minors is straightforward.
\end{proof}

We can now prove our main result.

\begin{proof}[Proof of~\Cref{thm_ain}]
Suppose by contradiction that $\minion Q_{\C^{2n}}\to\Pol(\Lin_n)$. Composing this minion homomorphism with  those from~\Cref{prop_pol} and Lemma~$\ref{lem_complexification}$ yields $\minion Q_{\R^{2n}}\to\AIP_n$. We then deduce from~\Cref{prop_measure} that there exists a $\Z_n$-measure $\mu$ on the subspaces of $\R^{2n}$ satisfying $\mu(\R^{2n})=1$, contradicting~\Cref{lem_hjs}.
    Hence, $\minion Q_{\C^{2n}}\not\to\Pol(\Lin_n)$, and it follows from~\Cref{thm_characterisation_PT} that $\Lin_n$ is pseudo-telepathic over $\C^{2n}$.
\end{proof}

\section{Consequences}

We say that a structure $\Y$ is \textit{pseudo-telepathic} if it is pseudo-telepathic over $\C^d$ for some $d\in\N$. As an immediate consequence of~\Cref{thm_ain}, we obtain a complete classification of pseudo-telepathic structures in terms of the notion of \textit{bounded width}, which captures solvability of the (non-uniform) CSP parameterised by $\Y$ (in symbols, $\CSP(\Y)$) via
fixed levels of a local-consistency procedure~\cite{Feder98:monotone}, see also other equivalent characterisations in~\cite{AtseriasO19,tz17:sicomp,Kolaitis00:jcss,Bulatov08:dulaties,Barto16:sicomp}.

If $\Y$ has bounded width, it follows from~\cite{Barto16:sicomp} that
there exists a minion homomorphism $\minion M_{\SDP}\to\Pol(\Y)$, where $\minion M_{\SDP}$ is the minion, defined in~\cite{cz23soda:minions} and, independently, in~\cite{BrakensiekGS23}, capturing solvability via the basic semidefinite program of~\cite{Raghavendra08:everycsp}. Since $\minion Q_H\to\minion M_{\SDP}$ for each $H$ (as proved in~\cite[Theorem~24]{ciardo_quantum_minion}), in this case $\Y$ is not pseudo-telepathic. Alternatively, the latter fact can be derived from a result obtained independently in~\cite[Theorem~12]{bulatov2025satisfiability} for a formally different notion of non-commutative satisfiability extending the one in~\cite{AtseriasKS19} to the non-Boolean setting. It was observed therein that 
the straightforward translation between this notion and the framework of PVMs defining quantum strategies does not work for non-Boolean domains in arbitrary dimension. For finite-dimensional Hilbert spaces, a recent unpublished argument due to Karamlou~\cite{karamlou_2026_personal} proves that the two notions are equivalent.

Conversely, if $\Y$ has unbounded width, it was proved in~\cite{barto2009constraint} (see also~\cite{BOP18}) that $\Pol(\Y)\to\Pol(\Lin_n)$ for some $n\geq 2$. Hence, by combining~\Cref{thm_ain} and~\Cref{thm_characterisation_PT}, it follows that $\Y$ is pseudo-telepathic. 

The discussion above proves the following.

\begin{corollary}
\label{cor_unbounded_width}
    A structure $\Y$ is pseudo-telepathic if, and only if, it has unbounded width.
\end{corollary}

In particular,~\Cref{cor_unbounded_width} improves on~\cite[Theorem~1]{bulatov2025satisfiability}, which proves the analogous classification in the more permissive commuting-operator model, with the result of~\cite{slofstra2024operator} taking the role of~\Cref{thm_ain}. 
We also note that the explicit description of the perfect quantum strategies witnessing pseudo-telepathy of the $\Lin_n$ games whose existence is asserted in~\Cref{thm_ain} can be derived from the construction in~\cite{harding2005group}. It follows from~\cite{qassim2020classical} that these strategies cannot involve only tensor products of generalised Pauli observables for odd $n$.  
Pseudo-telepathy of $\Y$ is a necessary requirement for the undecidability of the quantum CSP parameterised by $\Y$; i.e., the decision problem of determining whether a perfect quantum strategy for the $\X,\Y$ game exists, given $\X$ as input. It is not a priori sufficient: turning classical soundness into quantum soundness greatly affects the nature of undecidability-preserving reductions; see, e.g.,~\cite{Ji,culf2024re,paddock2025satisfiability,CJM2025quantum,culf2026quantum,ji2021mip}. In particular, the complexity of quantum linear equations modulo $n$, known to be undecidable for $n=2$~\cite{slofstra2019set}, remains open in the general case. 

\section*{AI disclosure}
I used ChatGPT to search the literature for $\Z_n$-analogues of Gleason's theorem. This is how I found~\cite{harding2005group}. 

\addcontentsline{toc}{section}{References}
\printbibliography
\end{document}